\documentclass[11pt]{article}

\usepackage{booktabs} 
\usepackage[font=footnotesize,labelfont=bf]{caption}
\usepackage{subcaption}
\usepackage{mathtools}
\usepackage{geometry}
\usepackage{graphicx}
\usepackage{amsmath}
\usepackage{hyperref}
\usepackage[svgnames]{xcolor}
\usepackage{amssymb}
\usepackage[all]{xy}
\usepackage{csquotes}
\usepackage{tikz}
\usetikzlibrary{patterns}
\usetikzlibrary{matrix}
\usetikzlibrary{backgrounds} 
\usepackage{float}
\usepackage{multirow}

\usepackage{amsthm}

\newtheorem{theorem}{Theorem}[section]

\newtheorem{lemma}[theorem]{Lemma}
\newtheorem{proposition}[theorem]{Proposition}

\newtheorem{definition}[theorem]{Definition}

\newcommand{\shortversion}[1]{}

\newcommand{\s}{{s}}

\newcommand{\valueupper}{{\overline W}}

\usepackage{stackengine}

\allowdisplaybreaks

\title{
On the Computational Complexity of Mechanism Design\\ in Single-Crossing Settings
}

\author{Moshe Babaioff \thanks{Hebrew University of Jerusalem. Email: moshe.babaioff@mail.huji.ac.il.
}  \and Shahar Dobzinski \thanks{Weizmann Institute of Science. Email: shahar.dobzinski@weizmann.ac.il. Supported by ISF grant 2185/19 and BSF-NSF grant (BSF number: 2021655, NSF number: 2127781).}
\and Shiri Ron\thanks{Weizmann Institute of Science. Email:  shiriron@weizmann.ac.il. Supported by an Azrieli Foundation fellowship, ISF grant 2185/19, and BSF-NSF grant (BSF number: 2021655, NSF number: 2127781).}.
}

\date{\today} 
\begin{document}


\maketitle
\thispagestyle{empty}

\begin{abstract}
We explore the performance of polynomial-time incentive-compatible mechanisms in single-crossing domains. Single-crossing domains were extensively studied in the economics literature. Roughly speaking, a domain is single crossing if monotonicity characterizes incentive compatibility (intuitively, an algorithm is monotone if a bidder that ``improves'' his valuation is allocated a better outcome).
That is, single-crossing domains are the standard mathematical formulation of domains that are informally known as ``single parameter''.
In all major single-crossing domains studied so far (e.g., welfare maximization in various auctions with single-minded bidders, makespan minimization on related machines), the performance of the best polynomial-time incentive-compatible mechanisms matches the performance of the best polynomial-time non-incentive-compatible algorithms. Our two main results make progress in understanding the power of incentive-compatible polynomial-time mechanisms in single-crossing domains:
\begin{itemize}
\item  We provide the first proof of a gap in the power of polynomial-time incentive-compatible mechanisms and polynomial-time non-incentive-compatible algorithms in any single-crossing domain: we present an objective function in a single-crossing multi-unit auction 
for which there exists an incentive-compatible mechanism that exactly optimizes the objective function.
For this objective function, there is polynomial-time algorithm that provides an approximation ratio of  $\frac{1}{2}$,
yet no polynomial-time incentive-compatible mechanism provides a finite approximation (under standard computational complexity assumptions). 

\item The objective function used above is not natural. We show that to some extent this is unavoidable by providing a sweeping positive result for the most natural objective function in multi-unit auctions, that of welfare maximization. 
We present an 
incentive-compatible FPTAS  mechanism for every multi-unit auction with single-crossing domains.
This improves over the mechanism of Briest et al. [STOC'05] that only applies to the much simpler case of single-minded bidders.
\end{itemize}
\end{abstract}

\pagebreak

\section{Introduction}
Taking a bird's eye view of Algorithmic Mechanism Design, one can observe that each problem in the field is traditionally classified to one of two major categories: ``single parameter'' or ``multi-parameter''. As we later discuss in-depth, there is no consensus on which domains constitute a multi-parameter domain and which are not, nor a general agreement on the precise mathematical definition of each domain category. 
However, for the sake of getting an intuitive feel, an agent is a ``single parameter'' if a single number can represent its valuation. For example, some papers (e.g., \cite{lehmann-ocollaghan-shoham, archer-papadimitriou-talwar-tardos}) consider combinatorial auctions with single-minded bidders, where the parameter is the value of the set that the bidder is interested in. Other papers consider single-minded bidders in a multi-unit environment (e.g., \cite{Mualem-nisan,briest}). A more involved environment is scheduling with related machines, introduced by Archer and Tardos \cite{AT01}, in which 
that parameter of each machine is its speed. 
In each of these domains there exists an incentive-compatible mechanism that optimizes the objective function exactly, but this mechanism is NP-hard to compute.
Yet, the performance of the best computationally-efficient incentive-compatible mechanisms matches the performance of the best computationally-efficient non-incentive-compatible algorithms.

By no coincidence, the references above can be considered as a list of 
greatest hits of the early days of Algorithmic Mechanism Design. Indeed, the conventional wisdom is that single-parameter domains are ``easy'' in the sense that they allow for a plethora of incentive-compatible algorithmic constructions. Roughly speaking, the algorithm must obey a relatively easy-to-satisfy monotonicity condition to be implemented in an incentive-compatible way.
In contrast, in multi-parameter domains -- informally defined as all domains for which describing the valuation requires two numbers or more -- there is no analogous easy-to-work-with condition. The ``hardness'' of multi-parameter domains is supported by the relative lack of applicable algorithmic constructions, limiting characterizations (e.g.,\cite{roberts,lavi-mualem-nisan}), as well as by a series of computational impossibility results (e.g., \cite{D11,Dob16b, weinberg-comm-impossibility}). 

The algorithmic success of single-parameter mechanism design raises the question of whether polynomial-time incentive-compatible mechanisms in single-parameter domains are as powerful as their non-incentive-compatible algorithmic counterparts (when the objective function can be exactly optimized by an incentive-compatible mechanism that is not computationally efficient under standard complexity assumptions
\footnote{This paper considers worst-case analysis. However, the situation is quite different in the Bayesian world: a line of work that has culminated in \cite{dughmi-hartline-kleinberg-niazadeh} shows how to reduce any computationally efficient Bayesian welfare-maximizing algorithm to a comparable incentive-compatible mechanism. This demonstrates the ``easiness'' of mechanism design in Bayesian settings, as this reduction applies even to multi-parameter settings. In contrast, the impossibility results discussed above show that there are no similar reductions in our non-Bayesian setting.}). This is also the main question that we consider in this paper. 
We currently do not even have a candidate domain nor an objective function that might illustrate a separation between the two. 
In all relevant domains that were considered in the literature, the performance of polynomial-time incentive-compatible mechanisms matches the performance of polynomial-time non-incentive-compatible algorithms. 

\subsubsection*{Single-Crossing Domains}

Before we begin our journey in single-parameter domains, it seems appropriate to lay stable mathematical foundations for the discussion. The 
naive attempt to define 
single-parameter domains as ones in which the valuation can be represented by a single number \cite{wiki-single} does not hold water as, e.g., 
multiple numbers can be encoded by a single number.

We start by revisiting the basic mechanism design definitions. There are $n$ players, and a finite set of alternatives $\mathcal A$. Each player $i$ has a valuation function $v_i:\mathcal A \rightarrow \mathbb R$ that specifies the value of the player for every possible alternative.
The set of possible valuations of player $i$ is $V_i$.
A (deterministic, direct) \emph{mechanism} is a pair $(f,P)$ where $f:\prod_{i=1}^n  V_i\rightarrow A$ is a \emph{social-choice function} that specifies the chosen alternative, and $P:\prod_{i=1}^n  V_i\rightarrow \mathbb R^n$ outputs a vector that specifies the payment of
each player $i$. 
A mechanism $(f,P)$ implements a social-choice function $f$ in dominant strategies if for every two valuations $v_i, v'_i$ of player $i$ and valuations $v_{-i}$ of the other players, it holds that $v_i(f(v_i,v_{-i}))-P_i(v_i,v_{-i}) \geq v_i(f(v'_i,v_{-i}))-P_i(v'_i,v_{-i})$.
For concreteness, consider the following settings:
\begin{itemize}
\item ``Binary Single-Parameter Domains'': 
for each player $i$ there is a set of \textquote{good} alternatives $\mathcal A_i\subseteq \mathcal A$ and a number $t_i>0$ such that $v_i(a)=t_i$ if 
$a\in \mathcal A_i$,
whereas $v_i(a)=0$ otherwise.
The standard example for a binary single parameter domain is the domain of \emph{known single-minded bidders}. 

\item ``Linear Single-Parameter Domains'': for each player $i$ there is a vector $(c_1,\ldots, c_{\mathcal A})$ such that for each $v\in  V_i$ there is a  number $t_v$ such that for every alternative $a\in \mathcal A$, $v(a)=t_v\cdot c_a$. An example for a linear single-parameter setting is scheduling with related machines.

\item ``Function Based Domains'': order the valuations in each $V_i$ in some order. For each alternative $a\in \mathcal A$, there is
a function $f^i_a:\mathbb R \rightarrow \mathbb R$ such that $v(a)=f^i_a(t_v)$ for every $v\in V_i$, where $t_v$ is the index of $v$ in the order.
\end{itemize}

As their names hint,  the first two domains are intuitively ``single parameter'': the valuations of each player $i$ are defined by a single number that naturally defines an order $\succ^v_{i}$ on the valuations. It is well known \cite{Mye81,AT01} that in these domains it holds that $f$ is dominant-strategy implementable if and only if the following ``monotonicity'' property holds: for every player $i$ and every valuations $v_i \succ^v_i v'_i$, and for every valuation profile $v_{-i}$, it holds that $v_i(f(v_i, v_{-i})) \geq 
v_i(f(v'_i, v_{-i}))$. Payment formulas for functions $f$ that are implementable in these domains are also given in \cite{Mye81,AT01}.

We have to be more careful when considering the third domain. It is not hard to see that as written, even though $f^i_a$ receives a single parameter, without any restriction on $f^i_a$ 
its definition allows to encode all possible domains.
However, for some choices of the $f^i_a$'s, it makes sense to call the valuations ``single parameter''. For example, if there are two alternatives $a,a'$ with $f^i_a(x)=x^2$ and $f^i_{a'}(x)=\sqrt x$ and the value of the remaining alternatives is identically $0$, 
then one can verify that any monotone allocation function is implementable.
Therefore, the challenge is to understand the conditions on the $f^i_a$'s that lead to ``single-parameter'' implementability.


Luckily, this question was extensively studied in the economic literature and several similar conditions are considered and used in, e.g., \cite{GUESNERIE1984,laffontmartimort2009,HERMALIN2014}.
The ``right'' generalization of single-parameter domains is often referred to as the Single-Crossing property. 
The condition was first used by  \cite{mirrlees1971} in the context of finding optimal taxation policies, and by \cite{spence1973} in the context of the role of signaling in the job market, and thus it is also called Spence-Mirrlees property.

Specifically, in this paper we use the following formulation:
a domain $V_i$ is \emph{single-crossing} if there is a total order on the valuations $\succ^v_i$ and a total order on the alternatives $\succ^a_{i}$, such that for every two valuations ${v'}\succ^v_{i} v$ and two alternatives ${a}'\succ^a_{i} a$ it holds that ${v}'({a}')-{v}'(a)\geq {v}({a}')-{v}(a)$. 
When the single-crossing property holds, then indeed monotonicity (up to tie breaking) is equivalent to implementability. Note that binary single-parameter domains and linear domains are indeed single-crossing domains.

\subsubsection*{Multi-Unit Auctions}
The goal of this paper is to understand the performance of polynomial-time dominant-strategy mechanisms
in single-crossing domains: is the power of polynomial-time dominant-strategy mechanisms for single-crossing environments weaker than the power of polynomial-time algorithms that are not necessarily incentive-compatible? We focus on multi-unit auctions, a domain which often serves as an ideal playground for understanding the computational complexity of incentive-compatible mechanisms both in single-parameter domains \cite{Mualem-nisan,briest} and in multi-parameter domains \cite{lavi-mualem-nisan,dobzinski-nisan-limitations, dobzinski-nisan-roberts}. See also a survey by Nisan \cite{nisanhandbook2015}.

In a multi-unit auction there are $m$ identical items and $n$ players. Each player $i$ has a valuation function $v_i:[m]\rightarrow \mathbb R$. As usual, we assume that the valuations are non-decreasing and normalized ($v_i(0)=0$). The single-crossing property translates in multi-unit auctions to having some order $\succ_i$ on the valuations of each player $i$ such that for each $\overline v_i \succ_i v_i$ and $\overline s>s$ it holds that $\overline v_i(\overline s)- \overline v_i(s) \geq  v_i(\overline s)-v_i(s)$. The input consists of a number $t_i$ for each player $i$ that specifies that the valuation of player $i$ is the $t_i$'th valuation according to player $i$'s order $\succ^v_{i}$ and the number of items $m$. We assume that for each player $i$ we have access to extended value queries: given an index $t$ and  number of items $s$, an extended values query returns the value $v_i^t(s)$, where $v_i^t$ is the $t$'th valuation in the single-crossing domain $V_i$ of player $i$. Ideally, the running time is polynomial in the input size $poly(\log (\Sigma_i|V_i|), \log m, n)$.

\subsection*{Our Results}
\subsubsection*{Polynomial-Time Algorithms Beat Mechanisms in Single-Crossing Domains}

Our first result is a separation result: we show that there is a single-crossing setting in which the approximation ratios that  polynomial-time dominant-strategy mechanisms can provide are strictly  worse 
than the approximation ratios provided by polynomial-time algorithms. Such separations are known for various multi-parameter domains \cite{PSS08,D11,DV12,Dob16b, weinberg-comm-impossibility}
but our separation is the first in a simpler (and thus harder to separate) ``single-parameter'' setting. It is also important to note that with the exception of \cite{DV12}, the separations above are all in the communication setting, whereas our separation is in the Turing machine model.

To prove this result we take liberty in defining our objective function. Many objective functions were considered in the literature (e.g., welfare maximization, makespan minimization, various fairness notions). Some of them cannot be implemented by any dominant-strategy mechanism (e.g., Rawlsian welfare). 
More relevant to this paper are objective functions that can be implemented by a dominant-strategy mechanism, 
but not in a computationally efficient way, like welfare maximization, or makespan minimization for related machines.
Our main negative result is a separation result for single-crossing multi-unit auction domains:

\vspace{0.1in} \noindent \textbf{Theorem: }There exists an objective function for a multi-unit auction with two single-crossing bidders such that: 
\begin{itemize}
    \item There is an incentive-compatible mechanism that exactly optimizes the objective function. 
    \item There is a polynomial-time algorithm that provides an approximation ratio of $\frac{1}{2}$ to the optimum. 
    \item Every polynomial-time incentive-compatible mechanism that provides some finite approximation to the optimum satisfies that computing its allocation rule is TFNP-hard.\footnote{Recall that TFNP  stands for "Total Function Nondeterministic Polynomial". This is the class of search problems that are guaranteed to have an answer that can be verified in polynomial time.
    For example, every problem in the classes PPAD (including computing a Nash equilibrium) and PLS is also in TFNP. }
\end{itemize}

Although this is the first separation of its kind in single-crossing domains\footnote{Related is the paper \cite{brendan} that shows the impossibility of a black box that converts any algorithm to an incentive-compatible mechanism in binary single-parameter settings. But even this result comes short of pointing out a specific setting that is hard for polynomial-time mechanisms.}, we note that it is a synthetic separation. 
The objective function we use gives score of $2$ to welfare-maximizing allocations 
(we carefully construct our domain so that deciding whether a given allocation maximizes the welfare is easy), most other allocations get a score of $0$, except a few of them which get a score of $1$.

Maximizing the objective function can be done by an incentive-compatible mechanism that computes the welfare-maximizing allocation and uses VCG payments. However, computing a welfare-maximizing allocation is TFNP hard. A polynomial-time algorithm can overcome this barrier by selecting sometimes allocations with score $1$, but we show that in this case monotonicity is not preserved and thus the algorithm cannot be implemented in dominant strategies. Thus, a monotone polynomial-time algorithm must sometimes output an allocation with score $0$, and thus does not achieve any  finite approximation to the objective function.

\subsubsection*{An Incentive-Compatible FPTAS for  Multi-Unit Auctions with Single-Crossing Valuations}

Can a similar separation be achieved for some natural objective function? For makespan minimization with related machines a long line of work established that the approximation ratio that can be achieved by polynomial-time incentive-compatible mechanisms matches what can be achieved by polynomial-time algorithms (e.g., \cite{AT01,KOV05,DDDR11,CK10}).

The other well-studied goal is welfare maximization, that is, find an allocation $(\s_1,\ldots, \s_n)$ that (approximately) maximizes $\sum_{i=1}^n v_i(\s_i)$. For single-minded bidders, Briest et al. \cite{briest} provide an incentive-compatible FPTAS, which matches the best approximation that
can be achieved by a polynomial-time algorithm\footnote{For the case of multi-unit auctions with general valuations we have a $\frac 1 2$-approximation mechanism that makes polynomially many value queries \cite{dobzinski-nisan-k-minded}, as well as some evidence that this is the best possible ratio achievable by polynomial mechanisms \cite{dobzinski-nisan-roberts}. If randomization is allowed, there exists an incentive-compatible-in-expectation FPTAS \cite{DD09} and a universally incentive-compatible PTAS \cite{VOCKING2019}.}. However, domains of single-minded bidders are a relatively simple type of single-crossing domains. Ignoring incentives issues, there is a simple FPTAS to the welfare for all domains, including single-crossing domains. Exact welfare-maximization is known to be NP-hard, regardless of incentives (see ,e.g., \cite{nisanhandbook2015}). 

Our main positive result is an incentive-compatible FPTAS for welfare maximization for all single-crossing multi-unit auctions, implying 
that for welfare maximization in single-crossing domains, polynomial-time incentive-compatible 
mechanisms are as powerful as polynomial-time algorithms.\footnote{
Similarly to \cite{briest}, our result is in fact more general and a more careful analysis shows that it applies also to unknown single minded bidders, see Appendix \ref{sec-warmup}.}
In fact, we identify single-crossing domains as the richest set of settings currently known for which incentive-compatible deterministic mechanisms are as powerful as their randomized counterparts.

\vspace{0.1in} \noindent \textbf{Theorem  (informal): }There exists an incentive-compatible FPTAS for maximizing the welfare in multi-unit auctions with single-crossing bidders. 

\vspace{0.1in} \noindent
To get some intuition for our algorithm, let us first recall the standard (non-incentive-compatible) FPTAS. That FPTAS starts by rounding the values of each player to the nearest multiple of some $\delta>0$. It then applies dynamic programming to find an exact
welfare-maximizing solution with respect to the rounded values. If the valuations are single minded, then this results in an incentive-compatible mechanism as long as the choice of $\delta$ does not depend on the values of the players. 

However, for this algorithm to give an FPTAS, the parameter $\delta$ must be chosen at the appropriate scale 
so that the rounded values are close to the original values. The challenge is that different roundings result in different allocations. That is, if we simply use the value of the player with the highest value to determine the rounding parameter $\delta$, this player might win his favorite bundle with one value but might lose it when increasing his value, and the desired monotonicity is not obtained.

\cite{briest} cleverly solve this problem of finding the right $\delta>0$ by essentially ``trying'' all possible choices of $\delta$ and choosing the one that maximizes the welfare. They show that this algorithm is incentive-compatible and can be efficiently implemented. However, we do not know how to extend this result to all single-crossing domains. Thus we start by developing an alternative incentive-compatible FPTAS for single-minded bidders.


Our construction is quite simple: we let the rounding parameter depend on the value of the player with the highest value, but we give the player with the highest value a reward when computing the best outcome using the ``standard'' dynamic program. The higher the value of the player with the highest value, the bigger the reward is. Choosing the reward appropriately, we get that a winning player keeps his winning bundle even when the rounding changes.

Next, we extend this new monotone FPTAS for single-minded bidders to single-crossing $k$-minded bidders. A $k$-minded valuation of bidder $i$ has $k$ steps (quantities for which value might increase), with steps at quantities $s^i_1,\ldots, s^i_k$, where the $s^i_k$'s are publicly known. It is clear that if we want to somehow use the standard dynamic program we somehow have to round valuations so they will all be multiples of $\delta$. To understand how to obtain the right rounding, we spell out the following principle in the design of monotone algorithms: let $S$ be a subset of the domain. Applying an allocation rule that is monotone only with respect to $S$ on the full domain by replacing every $v$ with the maximal $v'\in S$ such that $v \succ v'$, results in a monotone algorithm with respect to the full domain. For example, for a given $\delta$ in the domain of single-minded bidders, $S$ is the set of valuations rounded down to the nearest multiple of $\delta$, and the allocation rule is selecting the welfare-maximizing allocation. 

We observe that rounding the $k$-minded valuations in the most straightforward way (i.e., round down each value to the nearest multiple of $\delta$) does not seem to be very helpful since the domain of rounded down valuations is in general not single crossing and thus applying a monotone algorithm on it does not guarantee incentive compatibility. Instead, we suggest to obtain a single-crossing domain of ``rounded down'' valuations by appropriately rounding down the marginal values. Now, the allocation rule that finds a welfare-maximizing allocation (after this novel rounding and after applying our reward scheme) is monotone. This results in an incentive-compatible FPTAS for single-crossing $k$-minded bidders. Note that the running time of this algorithm depends on $k$.


The final step is obtaining an incentive-compatible FPTAS for general single-crossing domains. Ideally, we would want to prove that for $k=poly(\log m,\frac 1 \varepsilon)$, every single-crossing domain of player $i$ admits a $k$-sketch: bundles $s^i_1,\ldots, s^i_k$ such that for every bundle $j$, $s^i_t \leq j < s^i_{t+1}$ and a single-crossing valuation $v$ in the domain, $v(s^i_t) \leq v(j) \leq (1+\varepsilon)\cdot v(s^i_t)$. The existence of such $k$-sketch effectively reduces general single-crossing domains to $k$-minded domains, for which we already have an algorithm. Unfortunately, there are single-crossing domains that do not admit $k$-sketches 
of size that is sub-linear in $m$ (see Appendix \ref{app:no-sketch}).

We overcome this barrier by considering a more nuanced variant of $k$-sketches, one in which a correct sketch is guaranteed only for bundles that have large enough values. We conclude the proof by showing that such a variant exists in every possible single-crossing domain and that this restricted variant indeed provides 
a $(1-\varepsilon)$-approximation by a monotone allocation rule.

\subsection*{Open Questions}
We conclude the introduction with some open problems. One result of this paper is that there exists a problem in a single-crossing domain in which the approximation ratio achievable by polynomial-time dominant strategy mechanisms is strictly worse than the approximation ratio achievable by polynomial-time algorithms. However, the objective function used in this separation was unnatural. Perhaps the biggest question left open in this paper is to determine whether there {exists a single-crossing domain} 
for which a similar separation be achieved for the standard objective of welfare maximization. 

Understanding single-crossing combinatorial auctions might be the first step in this direction: is there a deterministic polynomial-time dominant strategy mechanism that achieves an approximation ratio of $O(\sqrt{m})$ for combinatorial auctions when the domains are single crossing?
Understanding other natural objective functions is interesting as well. For example, is there a deterministic dominant strategy PTAS for minimizing the makespan with related machines when the domains are single crossing? 

\subsubsection*{Structure of the Paper}

In Section \ref{sec-preliminaries} we give some necessary preliminaries. In Section \ref{sec-alg-k} we provide an algorithm for single-crossing $k$-minded bidders. Section \ref{sec-alg} gives an algorithm for all single-crossing domains and in Section \ref{sec-hard-new-obj} we prove a gap in the power of polynomial-time incentive-compatible mechanisms and polynomial-time algorithms. {
In Appendix \ref{sec-warmup} we present a different monotone FPTAS
for single-minded bidders than the one obtained by \cite{briest}.} 

\section{Single Crossing Domains: The Basics} \label{sec-preliminaries}



We assume that there are $n$ players and a finite set of alternatives $\mathcal A$.  Each player $i$ has a domain of valuations $V_i$, where an element of $V_i$ is a valuation $v_i:\mathcal A\to \mathbb{R}$. A (direct, deterministic) \emph{mechanism} $M=(f,P_1,\ldots,P_n)$ is composed of a social-choice function $f:V_1\times \cdots \times V_n \to \mathcal{A}$  and $n$ payment schemes $P_1,\ldots,P_n: V_1\times \cdots \times V_n \to \mathbb{R}$. 
A mechanism $M=(f,P_1,\ldots,P_n)$ is \emph{incentive-compatible} if for every player $i$, for every $v_i,v_i'\in V_i$, and for every $v_{-i}\in V_{-i}$ it holds that  
$
v_i(f(v_i,v_{-i}))
-P_i(v_i,v_{-i}) \ge
v_i(f(v_i',v_{-i}))
-P_i(v_i',v_{-i})
$.

A function $f_i:V_i\to \mathcal A$ is \emph{implementable for player $i$} if there exists a payment function $P_i:V_i\to \mathbb{R}$ such that  $
v_i(f_i(v_i))
-P_i(v_i) \ge
v_i(f_i(v_i'))
-P_i(v_i')
$ for every $v_i',v_i\in V_i$. 
Similarly, a social-choice function $f$ is \emph{implementable} if for every player $i$, and for every $v_{-i}\in V_{-i}$, the function $f(\cdot,v_{-i}):V_i\to \mathcal A$ is implementable for player $i$.



Fix some player $i$ and let $\succ^v_i$ be some total order on the valuations in $V_i$ and $\succ^a_i$ be some total order on the alternatives. A function $f_i:V_i\to \mathcal A$ is \emph{monotone for player $i$} with respect to $\succ^v_i$ and $\succ^a_i$ if for every $v_i'\succ^v_i v_i$ it holds that 
$f_i(v_i')\succeq^a_i f_i(v_i)$. 
A social-choice function $f:V_1\times \cdots \times V_n \to \mathcal A$ is \emph{monotone} if for every player $i$, and for every $v_{-i}\in V_{-i}$, the function $f(\cdot,v_{-i}):V_i\to\mathcal A$ is monotone for player $i$ with respect to some orders  $\succ^v_i$ and $\succ^a_i$.

\begin{definition}
Let $V_i$ be a set of valuations, $\succ^v_i$ an order on $V_i$, and $\succ^a_i$ an order on the alternatives.   
The domain $V_i$ is a \emph{single-crossing domain} with respect to  $\succ^v_i$ and $\succ^a_i$, if for every $v_i'\succ^v_i v_i\in V_i$ and for every $a'\succ^a_i a\in \mathcal A$ it holds that  $v_i'(a')-v_i'(a)\ge v_i(a')-v_i(a)$. 
\end{definition}
Note that our definition is very similar but slightly more permissive than other definitions, e.g., of \cite{HERMALIN2014} (namely, weak inequality and not a strict one).
We observe that using our definition, single-crossing domains are exactly the set of domains where ``monotonicity implies implementability".
Towards this end, we recall that in an incentive-compatible mechanism, by the taxation principle, each alternative has price for each player that does not depend on the valuation of the player. Thus, we slightly abuse notation and use $P_i(a, v_{-i})$ to denote the price of player $i$ for alternative $a$, given the valuations $v_{-i}$ of the other players. In what follows, we use $P_i(a)$ when $v_{-i}$ is clear from the context. 
Proofs from this section (which are similar to other proofs that appear in the literature) can be found in Appendix \ref{sec-missing-proofs-prem}:
\begin{definition}
Fix some player $i$ and let $V_i$ be his set of valuations. 
Let $\succ^v_i$ be some total order on $V_i$ and $\succ^a_i$ be some total order on the alternatives. $V_i$ is a \emph{monotone-implementability} domain
with respect to $\succ_i^v$ and $\succ_i^a$,
if every
function that is monotone for player $i$ with respect to  $\succ^v_i$ and $\succ^a_i$ is also implementable for player $i$.
\end{definition}

\begin{proposition}\label{prop-mono-iff-sc}
    Fix some player $i$ and let $V_i$ be his set of valuations. 
Let $\succ^v_i$ be some total order on $V_i$, and $\succ^a_i$ some total order on the alternatives. 
     $V_i$ is a monotone-implementability domain with respect to $\succ_i^{v}$ and $\succ_i^{a}$ if and only if it is also single-crossing with respect to $\succ_i^{v}$ and $\succ_i^{a}$.
     Moreover, let $f$ be a monotone social-choice function and let $a_j$ denote the $j$'th alternative (according to $\succ^a_i$) in the image of $f$. A payment function that implements $f$ is given by:
     $$
     P_i(a_j)-P_i(a_{j-1})=\inf_{v_i\in V_i}\{v_i(a_j)-v_i({a_{j-1}})\hspace{0.2em}|\hspace{0.2em} f(v_i)=a_j\}
     $$
\end{proposition}
The set of functions that can be implemented in a monotone-implementability domain is slightly larger than the set of monotone functions. We will now describe this set of functions. 
Fix some player $i$ and let $\succ^v_i$ be some total order on the valuations in $V_i$ and let $\succ^a_i$ be some total order on the alternatives. A function $f:V_i\to \mathcal A$ is \emph{monotone up to tie-breaking} with respect to $\succ^v_i$ and $\succ^a_i$ if and for every 
{$v_i'\succ^v_i v_i$ and for every $a'\succ^a_i a$ such that $f(v)=a'$ and $f(v')=a$, it holds that $v'(a')-v'(a)=v(a')-v(a)$. 
}

\begin{proposition}\label{prop-tie-breaking-cool}
    Fix some player $i$ and let $V_i$ be a single-crossing domain with respect to some total order $\succ^v_i$ on $V_i$ and some total order $\succ^a_i$ on the alternatives. Let $f$ be an implementable function in $V_i$. Then, $f$ is monotone up to tie breaking with respect to $\succ_i^v$ and $\succ_i^a$. 
\end{proposition}

\section{A Monotone FPTAS for  $k$-Minded Single-Crossing Domains} \label{sec-alg-k}
In this section we present a deterministic incentive-compatible FPTAS for multi-unit auctions with single-crossing $k$-minded bidders. Its running time depends polynomially on $k$.  
In Section \ref{sec-alg} we provide an incentive-compatible FPTAS for all multi-unit auctions where the bidders have single-crossing valuations that are not necessarily $k$-minded.
The mechanism of Section \ref{sec-alg} uses (as a black box) this section's mechanism. 
We start with a formal definition.
\begin{definition}\label{def:k-minded}
	A valuation $v$ is \emph{$k$-minded} if there exists a $k$-tuple $K=(s_{1},\ldots,s_{k})$ such that $v(s)=v(s-1)$ for every $s\notin K$. 
A domain of valuations $V_i$ is \emph{$k-$minded} if 
all valuations in the domain are $k$-minded with respect to the same set of quantities $K$.
\end{definition}

In this section and Section \ref{sec-alg} we assume that for each player $i$ we have access to extended value queries: given an index
$t$ and a quantity $s$, return the value $v_i^{t}(s)$, where $v_i^{t}$ is the $t$'th valuation in the single-crossing domain $V_i$ of player $i$ (we may think about the oracle that answers these queries as a succinctly represented circuit that describes the single-crossing domain). 

The input of the algorithm is $(t_1,\ldots, t_n)$, where $t_i$ denotes the index of the valuation $v_i$ of agent $i$ in the single-crossing domain $V_i$, and the number of items $m$.
We also assume that (rational) numbers in the queries are given in some standard way such that the number of bits that it takes to represent some number $x$ is not larger than the number of bits it takes to represent any number $x'>x$.

\subsection{The Mechanism} \label{subsec-alg}


We now describe the allocation rule. We will show that it gives $(1-\varepsilon)$-approximation to the maximum welfare, and that it can be implemented by an incentive-compatible mechanism in polynomial time. For the description of the allocation rule, we need to introduce two technical components. The first one is a particular way of rounding valuation functions by rounding the marginals. 

\begin{definition}\label{def-delta-marginal}
Let $v$ be a valuation function. The \emph{$\delta$-marginal-rounded valuation} of $v$ is:
\begin{align*}
	v^\delta(s)=\left\{
	\begin{array}{ll}
    0 , & s=0; \\
    \sum_{j=1}^{s}\lfloor v(j)-v(j-1) \rfloor_{\delta}, & s\neq 0.
	\end{array}
	\right.
\end{align*}
\end{definition}
where $\lfloor x \rfloor_\delta$ denotes the value of $x$ rounded down to the nearest multiple of $\delta$.  Note that we 
round the \emph{marginals} of the valuation to multiples of $\delta$, and not the values. 

In addition, the algorithm outputs an allocation of the items that maximizes the welfare in some instance. If there are several such welfare-maximizing allocations, we choose one of them according to the following order that we use as a tie-breaking rule: for every two different allocations $s'=(s_1',\ldots,s_n')$, $s=(s_1,\ldots,s_n)$ satisfy  that $s'\succ s$ if $\sum_{j=1}^{n}s_j'$ is smaller than $\sum_{j=1}^{n}s_j$.
If $\sum_{j=1}^{n}s_j=\sum_{j=1}^{n}s_j'$, then $s'\succ s$ if ${s}'_l>s_l$, where $l$ is the largest index where 
$s$ and $s'$ differ. The tie breaking rule chooses the largest allocation according to $\succ$ among all the welfare-maximizing allocations.

We are finally ready to describe the allocation rule. Fix $\varepsilon>0$. The allocation rule is as follows: 

\begin{enumerate}
\item Let
 $\delta=\max_{ p\in \mathbb{Z}} \big\{(4kn)^p| (4kn)^p \le \frac{\varepsilon \cdot v_{max}}{3n^2k^2}\big\}$ where $v_{max}=\max_{i}v_i(m)$.  

\item\label{alg-step-rounding}  For every player $i$, let $v_i^\delta$ to be the $\delta$-marginal-rounded valuation of $v_i$.

\item\label{alg-step-rewards} 
Let $TOP = \big\{i\in [n] \hspace{0.25em} \big| \hspace{0.25em} v_i(m) \ge \frac{3\delta n^2k^2}{\varepsilon} \big\}$. For each player $i\notin TOP$, let $u_i^\delta=v_i^\delta$.
For each player $i\in TOP$ which is $k_i$-minded with respect to $K_i=(k^i_1,...,k^i_{k_i})$, define $u_i^\delta$ by
$$u_i^\delta(s)=v_i(s)+|\{j|k^i_j\leq s\}|\cdot 2\delta kn $$ 
{where $k=\max_i |K_i|$.}

\item \label{alg-step-dynamic} Find a welfare-maximizing 
allocation for the instance $(u_1^\delta,\ldots, u_n^\delta)$ and allocate accordingly. If there are several optimal allocations, choose one using the tie breaking rule described above.
\end{enumerate}

 \begin{theorem}\label{thm-k-alg}
 Let $c$ be the number of bits used to represent numbers in an instance $(v_1,\ldots, v_n)$. The allocation rule above is implementable and outputs an allocation that is a $(1-\varepsilon)$ approximation to the maximal welfare. The allocation rule and its associated payments can be computed in time $poly(k,n,\frac{1}{\varepsilon},\log m, c,\log (\max_i t_i))$ with the same amount of extended value queries.
 \end{theorem} 
A formal proof can be found in Subsection \ref{subsec-k-alg-proof}, but let us now give a brief informal proof sketch. Assume for simplicity that all values are integers. Our plan is to solve the problem using the standard dynamic program (see, e.g., \cite{nisanhandbook2015}) which computes the optimal solution in time $poly(n,\log m, W)$,
where $W$ is at least the value $OPT$ of the optimal solution. 
The problem is that $W$ might be large and in any case it is not clear how to obtain an estimation of $W$ in an incentive-compatible way so we cannot directly apply this dynamic program.

Assume for now that $W$ is known. Ignoring incentives, the standard algorithm rounds down each value to the nearest multiple of some $\delta$. Finding the optimal solution with the rounded down valuations now takes time $poly(n,\log m,\frac W \delta)$ and gives a solution with value at least $OPT-n\cdot \delta$, which is a good approximation when, e.g.,  $OPT\leq W\leq n\cdot OPT$ and $\delta=\frac {\varepsilon\cdot W} {n^2}$. This rounding technique works well for binary single-parameter domains, e.g., single-minded bidders, simply because running any monotone allocation rule on the rounded valuations preserves the monotonicity of the allocation rule. Indeed, it yields an incentive-compatible FPTAS for single-minded bidders when $W$ is known. 


When the valuations are $k$-minded, it is less obvious how to correctly round the valuations. Thus, we observe the following general rounding technique for single-crossing domains: select a subset of the valuations $R$ from the domain.``Round down'' each valuation $v$ to the largest valuation in $R$ that precedes it in the order that is specified by the single-crossing domain. The key observation is that running a monotone allocation rule (e.g., the dynamic program mentioned above) on the rounded down valuations preserves monotonicity. 
The set $R$ that we use is the set of $\delta$-marginal-rounded valuations\footnote{In fact, the set $R$ does not have to be a subset of the original single-crossing domain, e.g., in our case the domain might not contain some $\delta$-marginal-rounded valuations. However, our allocation rule is monotone with respect to the union of the original domain and $R$, so monotonicity is still preserved.} (Definition \ref{def-delta-marginal}). 
Note that every valuation is well approximated by its $\delta$-marginal-rounded down version, thus the optimal welfare in any instance is very close to the optimal welfare with respect to the instance when each valuation is ``rounded down'' to a valuation in $R$. 
Also, the optimal allocation in instances where each valuation is $\delta$-marginal rounded down can be computed in time $poly(n,\log m,\frac W \delta,k)$ by the dynamic program, since we have an upper bound of $W$ on the value of the optimal solution and all values in the set $R$ are multiples of $\delta$. Hence we get a good incentive-compatible approximation algorithm as long as $W$ is known. 

It is left to find a good estimate $W$ on $OPT$. When bidder $i$ changes his valuation from $v_i$ to $\overline v_i$ such that $\overline v_i\succ v_i$, the allocation of the bidder does not decrease if $W$ does not change. Thus, we only need to take care of the case where $W$ does change. Suppose that our estimation $W$ changes and thus the rounding parameter changes from $\delta$ to $\overline\delta$. Denote by $(o_1,\ldots, o_n)$ the optimal solution in the instance $(v_i,v_{-i})$. We have two competing forces that determine the number of items player $i$ gets in the instance $(\overline v_i,v_{-i})$. On one hand, when we compare the welfare
of allocations in the instance $(v_i,v_{-i})$ to their welfare
in $(\overline v_i,v_{-i})$, the increase in the value of every allocation in which player $i$ gets $t$ items is at least the increase in the value of every allocation where player $i$ gets less than $t$ items, because of the single crossing property. This makes player $i$ more likely to win more items in $(\overline v_i,v_{-i})$ than in $(\overline v_i,v_{-i})$. However, the change in the rounding parameter might cause 
the (rounded) values of the other players to go down for some allocations where player $i$ wins many items, and the (rounded) values in some allocations where player $i$ wins less items might remain unaffected.
The result might be that another allocation (which has values that are less affected by the rounding) is the optimal solution in the instance $(\overline v_i,v_{-i})$. If player $i$ gets less than $o_i$ items in this new optimal allocation, monotonicity is violated. The reward given to player $i$ in Step \ref{alg-step-rewards} ``compensates'' on the loss due to the new rounding, and guarantees that player $i$ gets at least $o_i$ items.

\subsection{Proof of Theorem \ref{thm-k-alg}}\label{subsec-k-alg-proof}

We now prove Theorem \ref{thm-k-alg}. We start by proving that the allocation rule can be efficiently implemented. We then analyze the approximation ratio of the algorithm and prove that the allocation rule is implementable. Finally, we observe that the payments can be computed in polynomial time.

\subsubsection{The Complexity of Implementing the Allocation Rule}

The only step in the algorithm that is non-trivial to implement is the last one, which outputs a welfare-maximizing allocation with respect to valuations $u_1^\delta,\ldots,u_n^\delta$. We use a standard dynamic program for that (see, e.g., \cite{nisanhandbook2015}) with our specific tie breaking rule.

Let $\valueupper=n\cdot \max_i u_i^\delta(m)$. We construct a table $T$ of size $n\times (\valueupper/\delta)$
such that $T(i,w)$ contains the minimum number of items
such that allocating them optimally between players $1,\ldots,i$ with valuations $u_1^\delta,\ldots,u_i^\delta$ yields welfare of at least $w$.  {As standard, we use pointers to track the allocation that corresponds to each entry in the table.} 
Formally, we use the following recursive formula:
 
 \begin{gather*}
\forall w\in  \{0,\delta,\ldots,\valueupper\},\quad 
     T(1,w)=\arg\min_{s\in [m]} \{u_1^\delta(s)\geq w\}
\\
{\forall w\in  \{0,\delta,\ldots,\valueupper\},\hspace{0.25em} \forall i\in \{2,\ldots,n\}, \hspace{0.25em}
     T(i,w)= \min_{0\le j \le \frac{w}{\delta}} 
\big\{\arg\min_{s\in [m]} \{u_i^\delta(s)\geq j\delta\}+T(i-1,w-j\delta)\big\}}
 \end{gather*}
After we construct the table, we output the allocation with the highest welfare among the allocations in column $n$ that allocate at most $m$ items. 

Correctness is almost immediate, by recalling that for valuations $u_1^\delta,\ldots,u_n^\delta$ the welfare of every allocation is a multiple of $\delta$,  and that $\valueupper$ is an obvious upper bound on the optimal welfare. Note that as described, the program only finds the value of the maximal welfare.


{To implement the tie breaking rule, in each cell $T(i,w)$ we choose the largest allocation according to the order $\succ$, among the allocations that we consider in the \textquote{min} expression. We also implement the pointers as follows. If we allocate $x_i$ items to player $i$ in the cell $T(i,w)$, then this cell points to the cell $T(i-1,w-v_i(x_i))$. A simple induction gives that the tie-breaking rule mentioned above is implemented correctly:}
\begin{lemma}
    For every value $w$ and for every player $i$, let $s$ be the allocation that the cell $T(i,w)$ is associated with. Then, every allocation $\hat{s}$ that allocates to the players $1,\ldots,i$ and has welfare at least $w$ satisfies that $s\succ \hat s$. 
\end{lemma}
\begin{proof}
{We prove the lemma by induction}.
The claim trivially holds for all cells in the first column (for $T(1,\cdot)$).
The strong induction hypothesis says that every cell $T(j,w')$ where $j\le i-1$
is associated with the largest allocation according to the order $\succ$ that allocates to the player $1,\ldots,i-1$ and has welfare at least $w'$.

Fix a cell $T(i,w)$ in column $i$, and denote with $s=(s_1,\ldots,s_i,\emptyset,\ldots,\emptyset)$ the largest allocation according to the order $\succ$ that has welfare at least $w$. 
Assume {in contradiction} that $T(i,w)$ is associated with an allocation $\hat s=(\hat s_1,\ldots,\hat s_i)$ such that $s\succ \hat s$. Therefore, $\sum_{j=1}^i s_j=\sum_{j=1}^i \hat s_j$ ($\sum_{j=1}^i s_j\le \sum_{j=1}^i \hat s_j$ 
because $s \succ \hat{s}$, and the correctness of the dynamic program guarantees that $\sum_{j=1}^i s_j\ge \sum_{j=1}^i \hat s_j$). Since $s\succ \hat{s}$, $s_l>\hat{s}_l$ where {$l\geq 2$} is the largest index where $s$ and $\hat s$ differ. 


We obtain a contradiction by analyzing the cell $T(l,w-\sum_{j=1}^{l+1} u_j^\delta(\hat s_j))$ which is associated by construction with $(\hat s_1,\ldots,\hat s_l,\emptyset,\ldots,\emptyset)$. 
We remind that by definition for every $j> l$, $s_j=\hat{s}_j$. Therefore, $\sum_{j=l+1}^{n}u_j^\delta(s_j)=\sum_{j=l+1}^{n}u_j^\delta(\hat s_j)$, so both $(s_1,\ldots,s_{l},\emptyset,\ldots,\emptyset)$ and $(\hat s_1,\ldots,\hat s_{l},\emptyset,\ldots,\emptyset)$ have welfare of at least  $w-\sum_{j=1}^{l+1}u_j^{\delta}(\hat s_j)$. In addition, it implies that
$\sum_{j=l+1}^{n}s_j=\sum_{j=l+1}^{n}\hat s_j$ so combining this with the equality above implies that  $\sum_{j=1}^{l}s_j=\sum_{j=1}^{l}\hat s_j$. By assumption, $s_l>\hat{s}_l$, so $(s_1,\ldots,s_l,\emptyset,\ldots,\emptyset)\succ (\hat s_1,\ldots,\hat s_l,\emptyset,\ldots,\emptyset)$ even though the cell $T(l,w-\sum_{j=1}^{l+1} u_j^\delta(\hat s_j))$ is associated with the allocation $(\hat s_1,\ldots,\hat s_l,\emptyset,\ldots,\emptyset)$, so we get a contradiction to the induction hypothesis, as needed.      
\end{proof}

Now for the analysis of the running time. The running time of the first two steps of the allocation rule is $poly(k,n,c,\frac{1}{\varepsilon})$, 
because writing a value requires at most $c$ bits. 
Note that our choice of $\delta$ 
guarantees that $4nk\delta \ge \frac{\varepsilon\cdot v_{max}}{3n^2k^2}$, so $\delta \ge \frac{\varepsilon\cdot v_{max}}{12n^3k^3}$. Therefore, the number of rows in the table is at most:  
$$
\frac{\valueupper}{\delta}\le  \frac{n\cdot[v_{max}+2k^2\delta n]}{\delta}=poly(k,n,\frac{1}{\varepsilon})
$$
The number of columns is $n$, so the total number of cells is $poly(k,n,\frac{1}{\varepsilon})$. The value of each cell can be computed in time $poly(k,n,\frac{1}{\varepsilon},\log m,c)$ (we use binary search to find the minimal number of items that gives value at least $w$).
Therefore, the total running time is 
$poly(k,n,\frac{1}{\varepsilon},\log m,c)$.

\subsubsection{The Approximation Ratio}

Let
 $ \s=(\s_1,\ldots,\s_n)$ be the allocation that the algorithm specified above outputs for $(v_1,\ldots,v_n)$, and let 
$o=(o_1,\ldots,o_n)$ be an optimal allocation. 
Denote $\sum_i v_i(\s_i)$ with $ALG$ and 
$\sum_i v_i(o_i)$ with $OPT$. When proving that $ALG\geq (1-\varepsilon)\cdot OPT$ we observe that for every quantity $s$: 
\begin{equation}
{0\le}v_i(s)-v_i^{\delta}(s) \le k\delta  \label{eq-rounding}
\end{equation}
This is because zero marginals are not affected by the rounding, and each of the $k$ marginals of $v$ that is affected decreases by at most $\delta$.  

\begin{align*}
ALG&=\sum_i v_i(\s_i)\\
&\ge \sum_i v_i^\delta(\s_i) \\
&\ge  \sum_i u_i^\delta(\s_i)  - 2n^2k^2\delta &\text{(by Step \ref{alg-step-rewards})}
\\
&\ge \sum_i u_i^\delta(o_i) - 2n^2k^2\delta &\text{($\s$ is optimal for $u_1,\ldots,u_n$)} \\  
&\ge \sum_i v_i^\delta(o_i) - 2n^2k^2\delta &\text{(by Step \ref{alg-step-rewards})}
\\
&\ge \sum_i v_i(o_i) - nk\delta  -2n^2k^2\delta &\text{(by Equation (\ref{eq-rounding}))} \\
&\ge OPT  -3n^2k^2 \cdot \frac{\varepsilon\cdot v_{max}}{3n^2k^2} &\text{(by the definition of $\delta$)} \\
&\ge (1-\varepsilon)\cdot OPT &(v_{max}\le OPT)
\end{align*}

\subsubsection{Monotonicity of the Allocation Rule}\label{subsec-monotone}

Fix a player $i$ and a valuation profile $v_{-i}$ of all other players. Let $\overline{v}_i\succ {v}_i\in V_{i}$. 
Let ${\s}=({\s}_1,\ldots,{\s}_n)$ be the allocation that the algorithm outputs for $({v}_i,v_{-i})$. Our goal is to show that player $i$ wins at least $s_i$ items in the instance $(\overline v_i,v_{-i})$. 

We denote with $\delta$ and with $\overline \delta$ the rounding parameters when the input of the algorithm is $(v_i,v_{-i})$ and $(\overline{v}_i,v_{-i})$, respectively. 
Let 
$(v_i^\delta,v_{-i}^\delta)$ 
be the valuations defined in Step \ref{alg-step-rounding} of the algorithm when the input is $(v_i,v_{-i})$, and let
$(\overline v_i^{\overline \delta},v_{-i}^{\overline \delta})$
be the valuations
when the input is $(\overline{v}_i,v_{-i})$. Similarly, let $(u_i^\delta,u_{-i}^\delta)$
and 
$(\overline u_i^{\overline \delta},u_{-i}^{\overline \delta})$
be the valuations defined in Step \ref{alg-step-rewards} for the two inputs. We will also use the notation $\overline v^\delta_i$ to denote the $\delta$-marginal-rounded valuation of $\overline v_i$.
We divide the analysis into two cases:

\paragraph{Case I: $\overline \delta> \delta$.}
Fix an allocation $\hat{s}=(\hat s_1,\ldots,\hat{s}_n)$ with $\hat{s}_i<s_i$. We will show that 
$\overline u_i^{\overline \delta}(s_i)+\Sigma_{j\neq i}u_j^{\overline\delta }(s_j)>\overline u_i^{\overline \delta}(\hat s_i)+\Sigma_{j\neq i}u_j^{\overline\delta }(\hat s_j)$, which implies that the algorithm outputs an allocation where player $i$ wins at least $s_i$ items. 


We start by observing that for any $\delta$, the domain of $\delta$-marginal-rounded valuations of $V_i$ is also single crossing and preserves the order of $V_i$. Therefore:
\begin{gather} \label{eq-preserve-sid} 
\overline v_i^{\delta} \succ v_i^{\delta} \implies  
\overline v_i^{\delta}(\s_i)-\overline{v}_i^{\delta}(\hat s_i)\ge  v_i^\delta(\s_i)-v_i^\delta(\hat\s_i) 
\end{gather}
 Also,  the allocation $\s$ maximizes the welfare in the instance $(u_i^\delta,u_{-i}^\delta)$. Hence: 
\begin{equation} \label{eq-a1-wins}
    \begin{multlined} 
        u^{\delta}_i(\s_i)+\sum_{j\neq i}u^{\delta}_j(\s_j) \ge u^{\delta}_i(\hat{\s}_i)+\sum_{j\neq i}u^{\delta}_j(\hat{\s}_j) \\
    \implies v_i^{\delta}(\s_i)+\sum_{j\neq i}v_j^{\delta}(\s_j) +2n^2k^2\delta 
    \ge v_i^{\delta}(\hat{\s}_i)+\sum_{j\neq i}v_j^{\delta}(\hat{\s}_j)
    \end{multlined}
\end{equation}
where the latter inequality holds due to Step \ref{alg-step-rewards} of the algorithm.
In addition, note that by (\ref{eq-rounding}), we have that for every player $j$:
\begin{equation}\label{eq-presrve-sid-rounded}
    v_j^\delta(s)-v_j^{\overline\delta}(s)\le v_j(s)-v_j^{\overline \delta}(s) \le k\overline \delta 
\end{equation}
Therefore:
\begin{align}
\overline{v}_i^{\overline \delta}(\hat{\s}_i) + \sum_{j\neq i}v_j^{\overline \delta}(\hat{\s}_j)&\le  
\overline{v}_i^{\delta}(\hat{\s}_i) + \sum_{j\neq i}v_j^{\delta}(\hat{\s}_j)
&\text{($\overline \delta$ is a multiple of $\delta$)} \nonumber\\
 &=   \overline{v}_i^{\delta}(\hat{\s}_i) -v_i^{\delta}(\hat{\s}_i)+v_i^{\delta}(\hat{\s}_i)+ \sum_{j\neq i}v_j^{\delta}(\hat{\s}_j) \nonumber\\
&\le \overline{v}_i^{\delta}({\s}_i) -v_i^{\delta}({\s}_i)+v_i^{\delta}(\hat{\s}_i)+ \sum_{j\neq i}v_j^{\delta}(\hat{\s}_j) &\text{(by (\ref{eq-preserve-sid}), {as $\hat{s}_i<s_i$)}} 
\nonumber \\
 &\le \overline{v}_i^{\delta}({\s}_i) -v_i^{\delta}({\s}_i)+v_i^{\delta}(\s_i)+\sum_{j\neq i}v_j^{\delta}(\s_j) +2n^2k^2\delta  &\text{(by (\ref{eq-a1-wins}))} \nonumber \\
 &= \overline{v}_i^{\delta}({\s}_i)+\sum_{j\neq i}v_j^{\delta}(\s_j) +2n^2k^2 \delta  \nonumber\\
 &\le \overline{v}_i^{\overline\delta}({\s}_i)+\sum_{j\neq i}v_j^{{\overline\delta}}(\s_j)+ nk\overline \delta +2n^2k^2 \delta  \nonumber &\text{(by (\ref{eq-presrve-sid-rounded}))} \\
  &< \overline{v}_i^{{\overline\delta}}({\s}_i)+\sum_{j\neq i}v_j^{\overline\delta}(\s_j)+ 2nk\overline\delta &(2nk\delta<\overline \delta)  \nonumber\\
  \implies \overline{v}_i^{\overline\delta}(\hat{\s}_i) + &\sum_{j\neq i}v_j^{\overline\delta}(\hat{\s}_j) < \overline{v}_i^{\overline\delta}({\s}_i)+\sum_{j\neq i}v_j^{\overline\delta}(\s_j)+ 2nk\overline  \delta  \label{eq-contradict} \nonumber
\end{align}
Since the rounding factor strictly increases when player $i$'s valuation is $\overline v_i$, in the instance $(\overline v_i,v_{-i})$ we have that the only player in the set $TOP$ is player $i$, and thus $u_j^{\overline \delta}(\hat{\s}_j)=v_j^{\overline\delta}(\hat{\s}_j)$ for every ${j\neq i}$.
Denote with $l,\hat l$ respectively the indices of $s_i,\hat s_i$ in the $k$-tuple $ K_i$.
By adding
 $2nk\hat l\cdot \overline \delta$ to both sides of the inequality above, and as $\overline u_i^{\overline \delta}(\hat{\s}_j)=\overline v_i^{\overline\delta}(\hat{\s}_j)  +\hat{l}\cdot 2\delta kn$, we get: 
\begin{align*}
\overline u_i^{\overline \delta}(\hat{\s}_i)+\sum_{j\neq i}u_j^{\overline \delta}(\hat{\s}_j) &< \overline{v}_i^{\overline \delta}({\s}_i)+\sum_{j\neq i}u_j^{\overline \delta}(\s_j)+ 2nk(\hat l+1)\overline \delta  \\
&\le 
\overline u_i^{\overline \delta}({\s}_i)+\sum_{j\neq i} {u_j^{\overline \delta}}({\s}_j)
&(\text{$s_i>\hat{s}_i$, so $l\ge \hat l +1$})
\end{align*}
Thus, the algorithm does not
output the allocation $\hat s$
given $(\bar v_i,v_{-i})$. Therefore, it necessarily outputs an allocation where player $i$ wins at least $s_i$ items, as needed.

\paragraph{Case II: $\overline \delta=\delta$.} Observe that if the valuation of player $i$ changes in a way that does not affect the rounding parameter $\delta$, then all the marginals in his $\delta$-marginal-rounded valuation either increase or do not change. We will show that in this case -- also because the tie breaking rule is fixed -- the number of items that the player gets 
does not decrease, so the allocation rule is monotone. 


Let $\overline s$ be the optimal allocation  that the algorithm outputs given $(\overline u_i^{\overline \delta},u_{-i}^{\overline \delta})$ (if there are several optimal allocations, then $\overline s$  is chosen according to the tie-breaking rule).
We will show that $\overline s_i\ge s_i$. Observe that $\overline u_i^{\overline \delta}( s_i)- \overline u_i^{\overline \delta}( \overline s_i) \ge u_i^{\delta}(s_i)-  u_i^{\delta}(\overline s_i)$, since the set of all $\delta$-marginal rounded valuations of $V_i$ is single crossing, $\overline \delta =\delta$ by assumption and  the rewards of Step \ref{alg-step-rewards} do not change the marginal values.


We also observe that since  the rounding parameter is identical in the instances $(v_i,v_{-i})$ and $(\overline v_i,v_{-i})$, then the sets of players that are in $TOP$ in both instances are identical, perhaps except player $i$. Therefore,  $u_j^{\delta}=u_j^{\overline \delta}$ for every
$j{\neq} i$. In addition, the allocation $\overline s$ is optimal for the instance $(\overline u_i^{\overline \delta},u_{-i}^{\overline \delta})$.
Together, assuming towards a contradiction that $s_i > \overline s_i$:
\begin{equation}\label{eq-s-wins}
      u_i^{\delta}(\overline s_i)-u_i^{\delta}(s_i)\ge \overline u_i^{\overline \delta}(\overline s_i)-\overline u_i^{\overline \delta}(s_i)  
    \ge \sum_{j\neq i}{u}_j^{\overline  \delta}( s_j) - \sum_{j\neq i}{ u}_j^{\overline  \delta}( \overline s_j) = \sum_{j\neq i}{u}_j^{\delta}( s_j) - \sum_{j\neq i}u_j^{ \delta}( \overline s_j) 
\end{equation}
I.e., in the instance $(u_i^\delta,u_{-i}^\delta)$, the welfare of the allocation $\overline s$ is at least the welfare of the allocation $s$. Since by definition $s$ maximizes the welfare in the instance $(u_i^\delta,u_{-i}^\delta)$,the allocations $s,\overline s$ have the same welfare in this instance. Therefore, by our tie breaking rule, $s\succ \overline s$, so the fact that the dynamic program outputs the allocation $\overline s$ given $(\overline u_i^{\overline \delta},u_{-i}^{\overline \delta})$ implies that:
\begin{equation*}\label{eq-s-good-modified}
    \overline u_i^{\overline \delta}(s_i)+\sum_{j\neq i}u_j^{\overline \delta}(s_j) < 
\overline u_i^{\overline \delta}(\overline s_i)+\sum_{j\neq i}u_j^{\overline \delta}(\overline s_j)
\end{equation*}
Therefore:
\begin{equation*}
    u_i^\delta(s_i)-u_i^\delta(\overline s_i) \le \overline u_i^{\overline \delta}( s_i)- \overline u_i^{\overline \delta}( \overline s_i) 
    < \sum_{j\neq i}u_j^{\overline \delta}(\overline s_j)-\sum_{j\neq i}u_j^{\overline \delta}(s_j)= \sum_{j\neq i}u_j^{ \delta}(\overline s_j)-\sum_{j\neq i}u_j^{\delta}(s_j)
\end{equation*}
which is a contradiction since $s$ maximizes the welfare in $(u_i^\delta,u_{-i}^\delta)$. Hence, $\overline s_i \ge s_i$, as needed. 


\subsubsection{Payment Computation}\label{subsec-pay-k} 
Assume that player $i$ is $k$-minded with respect to the tuple of quantities $K_i$. Note that we can always assume that the number of items a bidder is assigned is in $K_i\cup \{0\}$, since by ``rounding down'' the input to the nearest element in $K_i\cup \{0\}$ the value of the bidder does not change.

Thus, to compute the payment that is specified in Proposition \ref{prop-mono-iff-sc}, 
for every quantity $s_j\in K_i$
that is smaller than the number of items that player $i$ wins given the output of $f(v_i^{[t_i]},v_{-i})$,  
we find the valuation 
 $v_i^{j}$ such that  $f_i(v_i^{j},v_{-i})=s_j$ and the marginal $v_i^{j}(s_j)-v_i^{j}(s_{j-1})$ is minimal, by executing a binary search over the subset of valuations  $\{v_i^{[0]},\ldots,v_i^{[t_i]}\}\subseteq V_i$. 
 
 Observe that  payment computation depends only on valuations that precede $v_i$ in the single-crossing order. Thus, by assumption, the number of bits used to represent numbers in every query is at most $c$.
We have that the payment of the players can be computed with $poly(k,n,{c},\log (\max_i t_i))$ extended value queries in time $poly(k,n,\frac{1}{\varepsilon},\log m,c,\allowbreak\log (\max_i t_i))$.

\section{A Monotone FPTAS for all Multi-Unit 
Single-Crossing Domains}\label{sec-alg}
In this section we show that there is an incentive-compatible FPTAS for maximizing the welfare in every single-crossing multi-unit domain.
In particular, this improves over the FPTAS of  \cite{briest} that applies only to multi-unit auctions with single-minded bidders. We obtain the FPTAS by reducing any multi-unit single-crossing domain to $k$-minded single-crossing domain, for a small enough $k$. We then use the monotone FPTAS of Section \ref{sec-alg-k} as a black box to derive our result.\footnote{For unknown single-minded bidders, the reduction that we perform in this section does not preserve the monotonicity of the allocation rule. 
Fortunately, as detailed in Appendix \ref{sec-warmup}, the reduction is superfluous in this case.
}

As in Section \ref{sec-alg-k}, for each player $i$ we have access to extended value queries. The input is $(t_1,\ldots, t_n)$, where  $t_i$ denotes the index of valuation $v_i$ of agent $i$ in the single-crossing domain $V_i$.


    
	\begin{theorem} \label{thm-alg} 
	For every $\varepsilon>0$ there exists a deterministic incentive-compatible mechanism that provides an $(1-\varepsilon)$-approximation to the welfare in multi-unit auctions where the domain $V_i$ of each player $i$ is finite and single crossing. Both the allocation and the payments are computed in time  $poly(n,\frac{1}{\varepsilon},\log m,b)$
with $poly(n,\frac{1}{\varepsilon},\log m,b)$ extended value queries, 
where  $b$ is the maximal number of bits used to represent numbers in the domains $V_1,\ldots, V_n$.
	\end{theorem}
Note that the input size is $\sum_{i=1}^n\log (|V_i|)+\log m$. Thus, as long as $b$ is polynomial in the input size we get an FPTAS that runs in time that is polynomial in the input size and $\frac 1 \varepsilon$ (as well as the unavoidable terms, $n$ and $\log m$). In particular, in the ``standard'' case in which each possible value of a bundle is in the set of integers $\{1,2,3, \ldots, poly(m)\}$, we get an algorithm that runs in time $poly(n,\frac{1}{\varepsilon},\log m)$.

The main lemma of this section ``reduces'' a single-crossing domain to a single-crossing $k$-minded domain without losing too much in the approximation ratio. We start with a definition.

\begin{definition}
Let $K\subseteq [m]$ be a subset of quantities. For a valuation $v$,  let $v^{K}$ the following ($|K|$-minded) valuation: 
    $v^K(s)= \max \{v(x)\hspace{0.25em}|\hspace{0.25em} x\le s, \hspace{0.25em} x\in  K\}$.
\end{definition}

\begin{lemma}\label{lemma-k-optimal}
Fix $\varepsilon>0$ and let $V$ be a single-crossing multi-unit auction domain. Let $b$ be the maximal number of bits needed to represent a value in $V$. 
	Then, there exists a set $ K\subseteq [m]$ of size $poly(\frac{1}{\varepsilon},b)$ such that:
 \begin{enumerate}
     \item For every $v\in V$ and for every $s\in [m]$, $|v(s)-v^{K}(s)|< \varepsilon \cdot v(m)$.
     \item Let $V^{ K}=\{v^K \hspace{0.25em}|\hspace{0.25em} v\in V  \}$. Then, the domain $V^K$ is single crossing and satisfies that for every  $\overline{v},v\in V$:
     \begin{equation*}\label{eq-single-preserve}
    \overline{v} \succ v \implies \overline{v}^K \succ v^K     
     \end{equation*}
 \end{enumerate}
Moreover, $ K$ can be found in time $poly(\frac{1}{\varepsilon},\log m,b)$ by making $poly(\frac{1}{\varepsilon},\log m,b)$ extended value queries.
\end{lemma}
We prove Lemma \ref{lemma-k-optimal} in
Subsection \ref{subsec-reduce-k-minded}. We show how the lemma implies the theorem in Subsection \ref{subsec-all-alg}.


\subsection{Proof of Lemma \ref{lemma-k-optimal}} \label{subsec-reduce-k-minded}

We gradually construct $K$, initializing $K= \{0\}$.
Let $u_1$ be the first valuation in $V$ that is not identically zero. 
For this valuation $u_1$, we add to $K$ the
first quantity $s_1$ that satisfies that $u_1(s_1)>0$. Then, we proceed by adding to $K$ the smallest quantity $s_2$ such that $u_1(s_2)\ge (1+\frac{\varepsilon}{2})u_1(s_1)$, then we add to $K$ the smallest quantity $s_3$ such that $u_1(s_3)\ge (1+\frac{\varepsilon}{2})u_1(s_2)$, and so on.  Observe that since each number is rational and represented by $b$ bits, up until now $|K|\leq \log_{1+\frac{\varepsilon}{2}} 2^b=poly(\frac{1}{\varepsilon},b)$ quantities.

 Now, let $u_2$ be the smallest valuation such that $u_2(m)\ge(1+\frac{\varepsilon}{2})u_1(m)$. 
 Similarly to before, we add to $K$ the smallest quantity $s_1$ such that $u_2(s_1)>0$, then the smallest quantity $s_2$ such that $u_2(s_2)\ge (1+\frac{\varepsilon}{2})u_2(s_1)$, and so on. We then choose $u_3$ to be the {smallest} valuation such that $u_3(m) \ge(1+\frac{\varepsilon}{2})u_2(m)$, and proceed until we finish processing the entire domain  $V$.
 Denote with $U\subseteq V$ the valuations that we added quantities for.

Note that  $|U|\le \log_{1+\frac{\varepsilon}{2}} 2^b$. For each $u\in U$, we add at most $poly(\frac{1}{\varepsilon},b)$ quantities to $K$. Thus, when the process terminates, $|K|=poly(\frac{1}{\varepsilon},b)$. For the time and query bound, note that we apply binary search over $[m]$ to find each quantity $s_j$ given $s_{j-1}$, and similarly we also
apply binary search over the domain $V_i$ to find each $u_i$ given $u_{i-1}$. Since $V$ is single crossing, $|V|\leq 2^b$, so the running time of the binary search is $\mathcal O(b)$. \footnote{To see this, let $v' \succ v \in V$ be two valuations. Since the $v,v'$ are different and $V$ is single crossing, there is some $s$ such that $v'(s)>v(s)$. 
$v'(m)-v'(s)\geq v(m)-v(s)$ which implies that $v'(m)>v(m)$. Thus every valuation in $V_i$ gives a different value to the grand bundle. $|V_i|\leq 2^b$ follows since we use $b$ bits to represent values.}

We now prove that $K$ has the claimed properties. For the first property, fix some valuation $v\in V$ and $s\in [m]$. Let $u\in U$ be the largest valuation such that $v\succeq u$. 
By the single-crossing definition: $$v(m)-v(s)\ge u(m)-u(s) \implies v(m)-u(m)
\ge v(s)-u(s)$$ 
Therefore, since by construction $v(m)-u(m)\leq  \frac{\varepsilon}{2}\cdot u(m)$ we have that $v(s)-u(s)\leq  \frac{\varepsilon}{2}\cdot u(m)$ as well.

Next, we want to show that $u$ is \textquote{close} to $u^K$. Let $s^\ast \in K$ be the largest quantity in $K$ such that $s\ge s^\ast$. Note that such a quantity necessarily exists because $0\in K$.   
Note that by construction  $u^K(s)=u(s^\ast)$. Thus, by the construction we also have that:
\begin{equation*}\label{eq-used-close}
u(s)-u^K(s)=u(s)-u(s^\ast)\le \frac{\varepsilon}{2}\cdot u(s^\ast) \le \frac{\varepsilon}{2}\cdot u(m)
\end{equation*}
In addition, note that by definition, $v \succeq u$ implies that for every $x\in [m]$, $v(x)\ge u(x)$, so by construction $v^K(s)\ge u^K(s)$.
Therefore:
\begin{align*}
v(s)-v^K(s) &= v(s)-u(s)+u(s)-v^K(s) \\
&\le (v(s)-u(s))+(u(s)-u^K(s)) &(v^K(s)\ge u^K(s)) \\
&\le \varepsilon\cdot u(m) \\
&\le \varepsilon\cdot v(m)
\end{align*}

Finally, we observe that $V^K$ is a single-crossing domain that admits the same order as $V$. Consider two valuations $\overline v\succ v$ from the single-crossing domain and two quantities, $s$ and $r$. We want to prove that $\overline v^K(s)-\overline v^K(r)\ge v^K(s)-v^K(r)$. 
Let $s^K,r^K$ be the largest quantities in $K$ that are at most $s,r$, respectively. Note that:
$
\overline v^{K}(s)-\overline v^{K}(r)=
\overline v^{K}(s^K)-\overline v^{K}(r^K)=\overline v(s^K)- \overline v(r^K)
$ and that analogously 
$
v^{K}(s)-v^{K}(r)=
v(s^K)-v(r^K)
$. Observe that $\overline v(s^K)- \overline v(r^K) \ge v(s^K)- v(r^K)$ since $V$ is single-crossing.
Combining the aforementioned equalities gives the proof.

\subsection{Proof of Theorem \ref{thm-alg}} \label{subsec-all-alg}
\paragraph{The Mechanism}
For each player $i$, the mechanism obtains the set $K_i$ as guaranteed by Lemma \ref{lemma-k-optimal}
where we use 
the {accuracy} parameter $\frac{\varepsilon}{2n}$. 
We can now run any incentive-compatible $\alpha$-approximation mechanism $M$ for $(\max_i|K_i|)$-minded single-crossing valuations on the instance $(v_1^{K_1},\ldots,v_n^{K_n})$, allocate and charge accordingly to obtain an {incentive-compatible} mechanism with a comparable approximation ratio.\footnote{We only have to make sure that the number of items that player $i$ gets in $M$ is in $K_i$. This is easy to achieve by rounding down the allocation of each player to the nearest quantity in $K_i$. Note that the incentive compatibility and the approximation guarantee of $M$ are preserved.} Here, we use and analyze the mechanism of Theorem \ref{thm-k-alg} with a precision parameter $\frac{\varepsilon}{2}$. 


\paragraph{The Approximation Guarantee.}
Fix a valuation profile $(v_1,\ldots,v_n)$ and let $(\s_1,\ldots,\s_n)$ be the output of $M$. Let $(o_1,\ldots,o_n)$ be a welfare-maximizing allocation. 
Denote the welfare of $(\s_1,\ldots,\s_n)$ with $ALG$ and the welfare of $o_1,\ldots,o_n$ with $OPT$. 
Observe that: 
\begin{align*}
    ALG &=  \sum_i v_i(s_i) \\
    &\ge \sum_i v_i^{K_i}(s_i) \\
    &\ge (1-\frac{\varepsilon}{2})\cdot \sum_i v_i^{K_i}(o_i) &\text{(by Theorem \ref{thm-k-alg})} \\
    &\ge (1-\frac{\varepsilon}{2})\cdot \Big(\sum_i \big[v_i(o_i)-\frac{\varepsilon}{2n}\cdot v_i(m)\big]\Big) &\text{(by Lemma \ref{lemma-k-optimal})} \\
    &\ge (1-\frac{\varepsilon}{2})\cdot
    \Big(OPT -\frac{\varepsilon}{2n}\cdot n\cdot  OPT\Big) &(v_i(m) \le OPT) \\
    &\ge (1-\varepsilon)\cdot OPT 
\end{align*}
\paragraph{Incentive Compatibility.} Note that none of the sets $K_i$ depends on the valuation of player $i$, that each player $i$ can only be assigned a number of items that is in $K_i$ and that values of quantities that are not in $K_i$ do not affect the allocation. Thus, the incentive compatibility of $M$ for $k$-minded single-crossing bidders immediately implies that our mechanism is incentive-compatible as well.

\paragraph{Time and Query Complexity.} 
By Lemma \ref{lemma-k-optimal}, finding the set $K_i$ for every player $i$ takes time 
$poly(n,\frac{1}{\varepsilon},\log m,b)$
and $poly(n,\frac{1}{\varepsilon},\log m,b)$ extended value queries.    Given the set of quantities $K_i$, 
projecting  $v_i$ to $v_i^{K_i}$ requires at most $poly(n,\frac{1}{\varepsilon},b)$ value queries and $poly(n,\frac{1}{\varepsilon}, \log m,b)$ running time. 
Now, the mechanism that is specified in Theorem \ref{thm-k-alg}
takes $poly(n,\frac{1}{\varepsilon},\log m,b)$ time, so the  
total
time and query complexity of the allocation rule  are $poly(n,\frac{1}{\varepsilon},\log m,b)$.





\section{Algorithms Beat Mechanisms in Single-Crossing Domains}
\label{sec-hard-new-obj}

In this section we prove the first separation between the power of polynomial-time algorithms and polynomial-time incentive-compatible mechanisms in a single-crossing setting. We prove the result for the setting of single-crossing multi-unit auctions.  

Formally, an \emph{objective function} is a function  $obj:V_1\times \cdots \times V_n\times \mathcal A\to \mathbb{R}$ that takes as input a valuation profile and an alternative, and outputs the score of the alternative given this  profile of valuations. A social choice function $f:V_1\times \cdots \times V_n\to \mathcal A$ is an \emph{$\alpha$-approximation} to an  objective function $obj:V_1\times \cdots \times V_n\times \mathcal A\to \mathbb{R}$ if for every $(v_1,\ldots,v_n)\in V_1\times\cdots\times V_n$ , it holds that
$
obj(v_1,\ldots,v_n,f(v_1,\ldots,v_n)) \ge \alpha \cdot 
\max_{A\in \mathcal A} \{obj(v_1,\ldots,v_n,A)\}
$.

Our interest in this section is in objective functions that can be implemented by an incentive-compatible mechanism but not in polynomial time (under standard complexity assumptions), e.g., welfare maximization. We present an objective function for which 
there exist a $\frac 1 2$-approximation polynomial-time algorithm, yet every polynomial-time incentive-compatible mechanism does not provide a finite approximation ratio\footnote{The constant $\frac 1 2$ is arbitrary and can be replaced by any other function of the input.}.
The downside of our construction is that the objective function is not natural. This is unavoidable to some extent as for the most natural objective function -- welfare maximization -- this paper shows that there is no gap between polynomial-time algorithms and their incentive-compatible counterparts.

We base our result on the hardness of TFNP, the class of total search problems. The class of TFNP contains subclasses such as PPAD, PLS, and CLS, and in particular many problems in TFNP that are widely assumed to be hard, e.g., finding a Nash equilibrium of a game and integer factoring.


\begin{theorem}\label{thm-alloc-hard}
Let $T$ be a problem in $TFNP$. There exists an objective function in a multi-unit auction with two single-crossing bidders such that:
\begin{enumerate}
\item\label{part-monotone} There exists an incentive-compatible mechanism that exactly optimizes the objective function. 
\item\label{part-approx-non-monotone}   There exists a (non-incentive-compatible) polynomial-time algorithm that provides an approximation ratio of $\frac 1 2$ to the optimum. 
\item  \label{part-mono-poly-impos} 
Computing the allocation function of any incentive-compatible mechanism that provides a finite approximation to the optimum 
is at least as hard as computing $T$ (up to polynomial factors).
\end{enumerate}
\end{theorem}
\begin{proof}
Denote the length of the input of $T$ by $n$ and recall that by definition, for every input there is a witness of length $n^c$, for some fixed integer $c>0$. We consider multi-unit auction valuations on $m=2^{n^c}$ items. For $s>1$, we associate every bundle that contains $s$ items with a binary string in the natural way, by using the binary representation of the number $s$.
The domain of each player consists of $2^n$ valuations. The $x$'th valuation in the domain $V_A$ of Alice is:
\begin{align*}
	v_A^x(s)=\left\{
	\begin{array}{ll}
 0, & s=0; \\
	10^n+10\cdot x+4\cdot s \cdot x, &s=1;\\
    10^n+10\cdot x+4\cdot s \cdot x, &\text{$s>1$ and  the string associated with $s$ is not a witness for $x$};\\
    10^n+10\cdot x+4\cdot s \cdot x +1, & \text{$s>1$ and the string associated with $s$ is a witness for $x$}.
	\end{array}
	\right.
\end{align*}
The $x$'th valuation of Bob in the domain $V_B$ is:
\begin{align*}
v_B^x(s)=\begin{cases}
0, \quad s=0; \\
10\cdot x+4\cdot s \cdot x, \quad s\ge 1.
\end{cases}
\end{align*}

Note that the domains of Alice and Bob satisfy the single-crossing property, because for every $\overline v\succ v$ in either $V_A$ or $V_B$, and for every quantity $s\in \{1,\ldots,m\}$, $\overline v(s)-\overline v(s-1) \geq  v(s)-v(s-1)$. 

Our objective function is as follows. 
Consider an instance $(v_A^{a},v_B^{b})\in V_A\times V_B$. If $a> b$, the allocation that gives all items to Alice gets a score of $2$. If $b>a$, then the allocation that gives $m-1$ items to Bob and one item to Alice gets a score of $2$. If $a=b$, then any allocation $(s,m-s)$ where $s$ is a witness for $a$ gets a score of $2$ (by the totality of $T$, there is at least one such allocation). In addition, when $a=b$ the allocation where Bob gets all items gets a score of $1$. Every other allocation gets a score of $0$.

Note that an allocation has a score of $2$ if and only if it maximizes the welfare.
Hence, the VCG mechanism implements
our objective function in an incentive-compatible way, proving Part \ref{part-monotone}.

As for Part \ref{part-approx-non-monotone}, consider this algorithm: in an instance $(v_A^{a},v_B^{b})\in V_A\times V_B$, if $a> b$ allocate Alice all items, and if $b>a$ allocate Alice one item and Bob $m-1$ items. When $a=b$ allocate all items to Bob (and get a score of $1$). This allocation rule can be implemented in polynomial time but by Proposition \ref{prop-tie-breaking-cool} it
is not implementable because it is not monotone up to tie breaking (Alice's allocation in the instance $(v_A^{x},v_B^{x+1})$ is one item, but it is zero items in the instance  $(v_A^{x+1},v_B^{x+1})$, whereas $v_A^{x}(1)-v_A^{x}(0)<v_A^{x+1}(1)-v_A^{x+1}(0)$).

Finally, we prove Part \ref{part-mono-poly-impos}. 
{Let $f$ be
a social-choice function that provides a finite approximation to the objective function.} We will show that $f$ necessarily maximizes the objective function. We
assume for the sake of contradiction that there exist two valuations $(v^a_A,v^b_B)$ such that $f$ does not output an allocation with score $2$.  Since by assumption $f$ provides a finite approximation to the welfare it must output an allocation with score $1$ given $(v^a_A,v^b_B)$, so by construction, $a=b$. Note that
Bob must get all the items in this case.

We proceed with case analysis.
If $a>0$, then a valuation $v_A^{a-1}$ exists, so given the instance $(v^{a-1}_A,v^b_B)$ we have that $a-1< b$ and thus Alice must be allocated one item by $f$ for the approximation guarantee to hold. {However, by construction $v_A^{a}(1)-v_A^{a}(0)>v_A^{a-1}(1)-v_A^{a-1}(0)$, so $f$ is not monotone up to tie breaking. By Proposition \ref{prop-tie-breaking-cool}, it is not implementable.}   
  
Similarly, if $v_A^{a}$ is the smallest valuation in $V_A$, then a valuation $v_B^{b+1}$ necessarily exists. In this case, in the instance  $(v^{a}_A,v^{b+1}_B)$ we have that $a< b+1$ and thus Bob wins $m-1$ items, which is contradiction to
{implementability due to the same reasons.}
Therefore, every {incentive-compatible} algorithm that provides a finite approximation to the objective function must maximize the objective function. Observe that maximizing the objective function is at least as hard as finding a witness for $T$, because for every instance $x$ of length $T$, the optimal allocation for $(v_A^x,v_B^x)$ outputs an allocation $(s,m-s)$ such that $s$ is a witness for $x$. Thus, computing an incentive-compatible algorithm for the objective function  
is at least as hard as solving $T$, which concludes the proof.
\end{proof}
We can also use this proof to derive a proof of hardness of welfare maximization in single-crossing multi-unit auctions with only two players. In contrast, recall that the hardness result for single-minded bidders necessarily requires many players, since with single-minded bidders exhaustive search takes time that is exponential in the number of players. 

\begin{proposition}
    \label{cor-2-player-hard}
Finding the welfare-maximizing allocation for two bidders with valuations in a single-crossing domain is NP-hard.  
\end{proposition}
The proof is a straightforward adaption of the proof of Theorem \ref{thm-alloc-hard} by replacing the problem $T$ with any $NP$-hard problem and replacing the objective function with welfare maximization. 

\section*{Acknowledgements}
We thank Liad Blumrosen for the helpful discussions. Part of the work done while the authors were affiliated with Microsoft Research.

\bibliographystyle{alpha}
\bibliography{mainarxiv}

\appendix
\section{Another Monotone FPTAS for Single-Minded Bidders}\label{sec-warmup}
The problem of welfare maximization for single-minded bidders admits a monotone fully polynomial-time approximation scheme (FPTAS)
\cite{briest}.  
In this section we present another monotone FPTAS for the single-minded problem, 
illustrating one of the main ideas that we use in the approximation scheme for 
{single-crossing}
bidders.
In addition,
we remind that \cite{briest} have presented a  monotone FPTAS algorithm for a wide class of problems: multi-unit auctions, job scheduling with deadlines and more. The algorithm that we describe below applies to all those settings as well. 

Note that every single-minded valuation $v_i$ can be represented by a value parameter $x_i$ and a quantity parameter $q_i$ such that for every $s\ge q_i$ it holds that $v_i(s)=x_i$, and $v_i(s)=0$ for every $s<q_i$. 
We say that bidder $i$ is a \emph{known single-minded
bidder} if his domain of valuations $V_i$ satisfies that there is a single quantity $q_i$ which is the quantity parameter of all the valuations in $V_i$. We say that bidder $i$ is \emph{unknown single-minded bidder} if there exist valuations in $V_i$ that differ in their quantity parameters.    
We remark that 
the algorithm that we describe below does 
not only solve the problem for \emph{known} single-minded bidders, but rather solves the more general problem of \emph{unknown} single-minded  bidders.\footnote{Note that a domain of unknown single-minded bidders does \emph{not} necessarily satisfy the single-crossing property.} 

Let $(x_1,q_1),\ldots,(x_n,q_n)$ be the valuations of $n$ single-minded bidders. 
{For simplicity, we assume that the values of all players are integers in $\{0,1,\ldots,poly(m)\}$.} 
Given an allocation $(s_1,\ldots,s_n)$, we say that a player \emph{wins} if $s_i\ge q_i$, and otherwise we say that he \emph{loses}.


It well known that  rounding the valuations to 
multiples of 
$\delta=\dfrac{\epsilon\cdot v_{max}}{n}$, where $v_{max}=\max_i v_i(m)$ 
and then running a dynamic program to find the welfare-maximizing allocation yields a fully polynomial-time approximation scheme. However, it is not monotone: it could be that a player wins his desired quantity when his value is $x_i$, but when he increases it to be $\overline x_i>x_i$, then he no longer wins because of the effect that he has on the rounding parameter $\delta$ through the report of the maximal value. 

The monotone approximation scheme of \cite{briest} uses the following two  techniques to  handle this issue: the 
 social-choice function that they describe is the one with the maximum welfare among an infinite set of 
 {bitonic}\footnote{{Roughly speaking, bitonicity (originally defined by \cite{Mualem-nisan}) means that two conditions hold: first, if a losing player increases his value and still loses, then the welfare does not increase. If a winning player increases his value he still wins and the welfare does not decrease.}}
 social-choice functions, 
 {where} the bitonicity
 of each function in the set implies that the maximizer is  monotone. However, {it is not immediate how this social-choice function can be implemented in polynomial time.
To handle the issue of maximization over an infinite set,} they have shown that it suffices to iterate over a \textquote{small} subset of the set of 
{bitonic} functions to find the maximizer.

Our solution addresses the monotonicity problem of the allocation of the player with the highest value in a more straightforward manner.  We choose the rounding parameter $\delta$ to be the largest power of $4$ that is smaller than $\frac{\epsilon\cdot v_{max}}{3n}$. Then, we round down all valuations $v_1,\ldots,v_n$ to have values that are multiples of $\delta$, and we denote the rounded valuations  with $ v_1^{\delta},\ldots,v_n^{\delta}$. 

{We now define new valuations $u_1,\ldots,u_n$, in which the player with the maximum value, which we call $top$, wins a reward. if there is more than one such player, ties are broken in favor of the lowest index player. For player $top$, we set $u_{top}^\delta(q_{top})=v_{top}^\delta(q_{top})+2n\delta$. For every player $j$ other than $top$, we set $u_j=v_j^\delta$.}  Then, we output a welfare-maximizing allocation for $u_{1},\ldots,u_{n}$, using a standard \textquote{knapsack-like} dynamic program (see \cite{nisanhandbook2015}). 
{Similarly to \cite{briest}, we compute the \textquote{threshold} payment for every player $i$ with valuation $(x_i,q_i)$ by performing a binary search on the valuations with quantity $q_i$ and values below $x_i$.}

 Now, showing that this algorithm $(1-\epsilon)$-approximates the welfare and that its running time is $poly(n,\log m,\frac{1}{\epsilon})$ 
 is straightforward.
We now explain why it is also monotone. Fix a player $i$,  and the valuations $v_{-i}=(x_1,q_1),\ldots,(x_{i-1},q_{i-1}),(x_{i+1},q_{i+1}),(x_n,q_n)$ of players $1,\ldots,i-1,i+1,\ldots,n$. Assume that player $i$ wins his desired quantity $q_i$ if he bids $v_i=(x_i,q_i)$.
We denote the allocation 
$(v_i,v_{-i})$ with 
 $s^{win}$. 
Let $\overline v_i$ be a valuation such that $\overline v_i=(\overline x_i,\overline q_i)$, where $\overline x_i\ge x_i$ and $\overline q_i\le q_i$. 

Now, we want to show that player $i$ wins  his desired bundle $\overline q_i$ given $\overline v_i$. For illustration, we explain only the \textquote{interesting} case, where 
the rounding parameter changes to be $\overline \delta$ instead of $\delta$ where player $i$ bids $\overline v_i$. 
Note that the welfare of $s^{win}$
increases by at least $2n(\overline \delta - \delta)$ (because $s^{win}$ allocates to player $i$ his desired bundle)
and decreases by at most $n\overline \delta$ due to the increase in the rounding factor. We denote with $(u_i^\delta,u_{-i}^\delta)$ and
$(\overline u_i^{\overline \delta},u_{-i}^{\overline \delta})$ the valuations following the rounding an the rewards given $(v_i,v_{-i})$ and with $(\overline v_i,v_{-i})$ respectively.  
We remind that by construction $\overline\delta\ge 4\delta$, so:
\begin{equation}\label{eq-alloc-win-increase}
u_i^\delta(s^{win}_i)+\sum_{j\neq i} u_j(\s^{win}_j)
 < \overline u_i^{\overline \delta}(s^{win}_i) + 
    \sum_{j\neq i} u_j^{\overline \delta}(\s^{win}_j)
\end{equation}
I.e., the welfare of $s^{win}$ strictly increases when increasing the rounding parameter, because $\overline \delta \ge \delta$ implies that
the additive term  added to the valuation of the player with the maximum value \textquote{compensates} for the loss of value due to the increased rounding factor. 
Let $s^{lose}$ be some allocation where player $i$ wins less than $\overline q_i$ items. 
Observe that: 
\begin{align*}
\overline u_i^{\overline \delta}(s^{lose}_i) + 
\sum_{j\neq i} u_j^{\overline \delta}(\s^{lose}_j) 
 &\le 
\sum_{j\neq i} u_j^{\overline \delta}(\s^{lose}_j) &(u_i^{\overline \delta}(s_i^{lose})=0)\\   
&\le u_i^{\delta}(s^{lose}_i) + 
 \sum_{j\neq i} u_j^{\delta}(\s^{lose}_j) 
&\text{($\overline\delta$ is a multiple of $\delta$)}
    \\
&\le  u_i^{\delta}(s^{win}_i) + 
    \sum_{j\neq i} u_j^{\delta}(\s^{win}_j)
 &\text{($s^{win}$ maximizes the welfare given $(u_i^\delta,u_{-i}^\delta)$)} \\
&< \overline u_i^{\overline \delta}(s^{win}_i) + 
    \sum_{j\neq i} u_j^{\overline \delta}(\s^{win}_j) &\text{(by Equation (\ref{eq-alloc-win-increase}))}
\end{align*}
Therefore, given $(\overline u_i^{\overline \delta}, u_{-i}^{\overline \delta})$, the dynamic program outputs an allocation where player $i$ wins at least $\overline q_i$ items, which completes the proof.  








\section{Single-Crossings Domains Do Not Admit $k-$Sketches}\label{app:no-sketch}
In this section  we prove that valuations in single-crossing domains 
cannot be approximated by 
fixing any number of quantities smaller than $m$, and projecting the valuations on these quantities to get a $k$-sketch.
A $k$-sketch of a domain $V_i$ is a set 
of bundles $s_1^i,\ldots,s_k^i$ such that for every bundle  $s^i_t \leq j < s^i_{t+1}$ and a single-crossing valuation $v\in V$, $v(s^i_t) \leq v(j) \leq (1+\varepsilon)\cdot v(s^i_t)$. The construction we present below 
demonstrates that there exist single-crossing domains that do not even admit an $(m-1)$-sketch. 
Note that finding such a domain is straightforward if  we allow valuations that grow exponentially. For example,  even a domain that consists of the single valuation $v(s)=2^{s}-1$ does not admit any $(m-1)$-sketch. 
In contrast, in the construction presented below, all the values are rational numbers where both the numerator and denominator are integers of size at most $poly(m)$.
\begin{lemma}
There exists a single-crossing domain $V_i$ where all values are integers in the range $\{0,1,\ldots,poly(m)\}$  that does not admit any $(m-1)$-sketch. 
\end{lemma}
\begin{proof}
In the proof we demonstrate that leaving out one  coordinate always results in a loss of factor of at least $2$. In fact, the proof holds for any constant number, and our choice of $2$ is arbitrary. 
We gradually construct a domain of size $\frac{m}{\log m}$. We assume without loss of generality that $m$ is a power of $2$. 
The first non-zero valuation in it, $v^{[1]}$, is defined as follows:
\begin{gather*}
    \forall s\in \{1,\ldots,m-\log m\} \quad  v^{[1]}(s)= 0 \\
    \forall s\in \{m-\log m+1,\ldots,m\} \quad v^{[1]}(s)=2^{s-(m-\log m)} 
\end{gather*}
Note that the values of the last $\log m$ quantities 
are subsequent powers of two, so none of them can be left out from the sketch without losing a factor of $2$. Now, given a valuation $v^{[x-1]}$ we define $v^{[x]}$ as follows: 
\begin{gather*}
v^{[x]}(0)=0 \\
 \forall s \in \{1,\ldots,m-x\log m\}\quad v^{[x]}(s)=0 \\
    \forall s\in \{m-x\log m+1,\ldots,m-(x-1)\log m\} \quad v^{[x]}(s) =  2^{s-(m-x\log m)} \\
    \forall s \in \{m-(x-1)\log m + 1,\ldots, m\} \quad v^{[x]}(s)=v^{[x]}(s-1)+v^{[x-1]}(s)- v^{[x-1]}(s-1)
\end{gather*}
For an illustration of the domain $V_i$, see Table \ref{table:1}. Note that every valuation $v^{[x]}$  satisfies that the quantities in $\{m-x\log m +1,\ldots,m-(x-1)\log m\}$ have to be included the sketch. Therefore, all $m$ quantities have to be inside any sketch that gives an approximation better than $\frac{1}{2}$. In addition, it is easy to verify that the domain above is single crossing.  
\begin{table}[h]\scriptsize \centering 
\begin{tabular}{|| c|c|c|c| c| c| c |c |c |c |c  |c ||} 
 \hline \hline
&$0$ & \ldots & \ldots &\ldots & $m-2\log m+1$ &  \ldots & $m-\log m$ & $m-\log m +1$ & $m-\log m +2$ & \ldots &  $m$   \\ 
 \hline\hline
 $v^{[1]}$&  $0$ & \ldots &\ldots &\ldots & $0$  & \ldots & $0$ & $2$ & $4$ &  \ldots & $m$  \\ 
 \hline
  $v^{[2]}$ & $0$ &\ldots &\ldots & $0$ & $2$ & \ldots & $m$ & $m+2$ & $m+4$ &  \ldots & $2m$  \\
 \hline
 \vdots &   \vdots & \vdots  &\vdots  &\vdots&\vdots & \vdots & \vdots & \vdots & \vdots & \vdots &  \vdots \\
 \hline
  $v^{[\frac{m}{\log m}]}$ & $0$ & \ldots &\ldots & \ldots & $m(\frac{m}{\log m}-2)+2$ &  \ldots & $m(\frac{m}{\log m}-1)$ & $m(\frac{m}{\log m}-1)+2$ & $m(\frac{m}{\log m}-1)+4$ &  \ldots & $\frac{m^2}{\log m}$ \\  [2ex] \hline 
\end{tabular}
\caption{The table above specifies the valuations in the domain $V_i$, which is a domain that does not admit an $(m-1)$-sketch. The cell in row $t$ and column $s$ specifies the value of $v_i^{[t]}$, the valuation at index $t$, for $s$ items.  
The valuations are ordered according to the order $\succ_i^{v}$ that is associated with $V_i$.}
\label{table:1}
\end{table}
\end{proof}

\section{Missing Proofs} \label{sec-missing-proofs-prem}
In this section we present proofs that were missing from Section \ref{sec-preliminaries}. 
Before we begin, we repeat here a well known result that we will use in our proofs. 
For that, we remind that a function $f:V_i\to \mathcal A$ is \emph{weakly monotone for player $i$} if for every $v_i,v_i'\in V_i$, if $f(v)=a$ and $f(v')=a'$, it implies that $v'(a')-v'(a)\ge v(a')-v(a)$. It is easy to see \cite{lavi-mualem-nisan} that every function $f:V_i\to \mathcal A$ that is implementable for player $i$ is weakly monotone. 

\subsection{Proof of Proposition \ref{prop-mono-iff-sc}}
Let $V_i$ be  a single-crossing domain with respect to $\succ_i^a,\succ_i^v$. 
We begin by proving that every monotone  function $f$ can be implemented in a single-crossing domain.
We present a price for every alternative and then show that the allocation rule with these prices is an incentive-compatible implementation of $f$.

For every $a_j$ in the image of $f$ such that $j\geq 2$, set the difference of the prices to
$P(a_j)-P(a_{j-1})=\inf_{v_i\in V_i}\{v_i(a_j)-v_i({a_{j-1}})|\ f(v_i)=a_j\}$ (note that $P(a_1)$ can be arbitrary).
Therefore, for every two alternatives $a_l\succ^a_i a_k$ in the image:
$$
P(a_l)-P(a_{k})=\sum_{j={k}+1}^{l} \inf_{v_i\in V_i}\{v_i(a_j)-v_i(a_{j-1})\hspace{0.25em}|\hspace{0.25em} f(v_i)=a_j\}
$$ 
To see that the payment function $P$ indeed implements $f$, let $v$ be a valuation of player $i$ such that $f(v)=a$. Our goal is to show that $a$ is a profit-maximizing alternative of player $i$. Suppose that the profit of $v$ for some $a_j$ is less than its profit on $a_{j-1}$: $v(a_j)-P(a_j)< v(a_{j-1})-P(a_{j-1})$. By rearranging, $v(a_j)-v(a_{j-1})<P(a_j)-P(a_{j-1})= \inf_{v_i\in V_i}\{v_i(a_j)-v_i({a_{j-1}})|\ f(v_i)=a_j\}$. Since it holds that $v(a_j)-v(a_{j-1})<v_i(a_j)-v_i({a_{j-1}})$ for all valuations $v_i$ with $f(v_i)=a_j$, $v$ precedes these valuations in the single crossing order. By the monotonicity of $f$, $a_j \succ^a_i a$. Similarly, if $v(a_j)-P(a_j)> v(a_{j-1})-P(a_{j-1})$, then $a \succ^a_i a_j$. That is, $v(a_j)-P(a_j)\geq v(a_{j-1})-P(a_{j-1})$ for every $a_j$ such that $a\succ_i^a a_j$ and $v(a_j)-P(a_j)\leq v(a_{j-1})-P(a_{j-1})$ for every $a_j$ such that $a_j\succ_i^a a$. I.e., the profit series $\{v(a_j)-P(a_j)\}_j$ does not decrease up to alternative $a$ and does not increase afterwards, hence $a$ is one of the profit maximizing alternatives, as needed.

For the other direction,
let $V_i$ be a monotone-implementability domain with respect to $\succ_i^{a},\succ_i^{v}$. If $V_i$ is not a single-crossing domain with respect to  $\succ^v_i,\succ^a_i$, 
  then there exist two valuations $v_i' {\succ_i^v} v_i\in V_i$ and two alternatives $a'{\succ_i^a} a$ such that $v_i'(a')-v_i'(a)< v_i(a')-v_i(a)$.
    {Let $f$ be the following monotone function: $f(u_i)=a'$ for every valuation $u_i{\succ_i^v} v_i$, and for any other $u_i$ let $f(u_i)=a$. } 
    Note that $f(v_i)=a$ and $f(v_i')=a'$, so $f$ is not weakly monotone as $v_i'(a')-v_i'(a)< v_i(a')-v_i(a)$. Hence, $f$ is not implementable for player $i$, and thus $V_i$ is not a
    {monotone-implementability domain.} 

\subsection{Proof of Proposition \ref{prop-tie-breaking-cool}}
Fix two valuations $v_i'{\succ_i^v} v_i$ and $a'{\succ_i^a} a$ such that $f(v_i')=a$ and $f(v_i)=a'$. 
Since $V_i$ is single-crossing with respect to $\succ_i^{v}$ and $\succ_i^{a}$, it holds that $v_i'(a')-v_i'(a)\ge v_i(a')-v_i(a)$. $f$ is implementable, hence weakly-monotone, so $v_i(a')-v_i(a)\ge v_i'(a')-v_i'(a)$. We get the proof by combining the two inequalities.

\end{document}